\documentclass[11pt,reqno]{amsart}
\usepackage{graphicx,epsfig}
\usepackage{todonotes}
\usepackage[numbers,sort&compress]{natbib}
\usepackage{hyperref}
\usepackage{soul}
\usepackage[foot]{amsaddr}
\usepackage[final,color,notref,notcite]{showkeys}

\usepackage[charter]{mathdesign}

\definecolor{labelkey}{rgb}{0,.56,.7}

\newcommand{\seqnum}[1]{\href{http://oeis.org/#1}{\underline{#1}}}

\def\df{\overset{\mathrm{df}}{=}}
\newcommand{\nn}{\nonumber}

\newcommand*{\at}{@}

\def\a{\alpha}

\def\D{\Delta}

\def\bbZ{\mathbb{Z}}
\def\bbR{\mathbb{R}}

\def\msE{\mathsf{E}}
\def\msD{\mathsf{D}}

\begin{document}

\theoremstyle{plain}
\newtheorem{theorem}{Theorem}
\newtheorem{corollary}[theorem]{Corollary}
\newtheorem{lemma}[theorem]{Lemma}
\newtheorem{proposition}[theorem]{Proposition}

\theoremstyle{definition}
\newtheorem{definition}{Definition}
\newtheorem{example}[theorem]{Example}
\newtheorem{conjecture}[theorem]{Conjecture}

\theoremstyle{remark}
\newtheorem*{remark}{\bf Remark}

\title{On the number of nonnegative solutions of a system of linear Diophantine equations}

\begin{abstract}
We derive a closed expression for the number of nonnegative solutions of a certain system of linear Diophantine equations. The motivation comes from high energy physics where the nonnegative solutions play a crucial role in the perturbative calculation for a class of Lagrangians describing the interaction of an atom with a boson field or a non-linear interaction of boson fields among themselves (the so-called interacting $\phi^n$ models). The linear system can be solved and the nonnegative solutions enumerated but a closed expression for the number of solutions is preferable to counting the solutions. Interestingly, the problem led to a construction of a simpler linear Diophantine system whose nonnegative number of solutions turns out to be the magic constant.
\end{abstract}

\author{Kamil Br\'adler}

\email{kbradler\at uottawa.ca}

\address{Department of Mathematics and Statistics, University of Ottawa, Ottawa, Canada}



\thanks{The author thanks Damien Roy for discussions, Robin Chapman for pointing to the method of partition analysis and Maurice Rojas for a comment.}

\maketitle

\thispagestyle{empty}

\allowdisplaybreaks

\section{Introduction}\label{sec:intro}

Linear Diophantine equations and their systems are easy to solve. There are three possibilities: either a system has no solution, one solution or infinitely many solutions. The case of one solution can be thought of as a special case of infinitely many solutions. The method to distinguish the  particular cases is known~\cite[ch.~1]{nathanson2008elementary} and the issue can often be decided by inspection (by a heuristic search for at least one solution). This can be contrasted with the case of a general Diophantine equation, or its system, where the decision whether a solution exists belongs to hard problems. Focusing on the linear case from now on, if a system has infinitely many solutions it may be of an interest to investigate the total number of nonnegative solutions. The answer is necessarily a finite number. One such a system appeared in the author's recent work~\cite{bradler2016unitary}
\begin{subequations}\label{eq:Diophantine}
  \begin{align}
    2\a_{11}+\a_{12}+\a_{13}+\a_{14} & = \ell_1,\label{eq:Diophantine1} \\
    \a_{12}+2\a_{22}+\a_{23}+\a_{24} & = \ell_2,\label{eq:Diophantine2} \\
    \a_{13}+\a_{23}+2\a_{33}+\a_{34} & = \ell_3,\label{eq:Diophantine3} \\
    \a_{14}+\a_{24}+\a_{34}+2\a_{44} & = \ell_4,\label{eq:Diophantine4}
  \end{align}
\end{subequations}
where $\ell_i,\a_{ij}\in\bbZ_{\geq0}$ such that $\sum_i\ell_i$ is even. Its importance comes from the fact that it is closely related to counting the number of Feynman diagrams for a wide class of boson models in interacting quantum field theory. The linear equations in~(\ref{eq:Diophantine}) are one of those cases where for a given $\ell_i$ one can quickly find a solution and conclude that the number of solutions for $\a_{ij}$ is infinite. It is only slightly less obvious to see whether the system itself has zero or infinitely many solutions  (any of the four equations in~\eqref{eq:Diophantine} imposes a constraint on a solution for the remaining three equations).

The closed expression for the total number of \emph{nonnegative} solutions proved to be a pertinacious problem to pursue and we present its derivation. We simplify the system by considering $\ell_i=\ell$ (both even and odd) but, as will become clear, our counting (proof) strategy can be used to count the solutions for different $\ell_i$ if there is a need for it. Also, by setting $\a_{i4}=\ell_4=0,\forall i$ the number of nonnegative solutions of the resulting linear system is interesting on its own and turns out to be related to Floyd's triangle~\seqnum{A006003} and the row, column and diagonal sum of the normal magic square of order $\ell$ (called the magic constant). The problems related to linear Diophantine equations and their systems often appear in the theory of integer programming~\cite{schrijver1998theory}, lattice counting~\cite{beck2007computing} and combinatorics~\cite{stanley1997enumerative}. Typically, one is interested in finding the solutions of linear equations rather than counting them. As argued in~\cite{bradler2016unitary}, that is not a problem here. System~\eqref{eq:Diophantine} is simple enough so that all nonnegative solutions can be systematically listed. Alternatively, one can easily cast the system into the Smith normal form~\cite{smith1861systems} and get a generic expression for all solutions. But the Smith form does not seem to provide an easy way of counting the solutions.

There exist several algorithms for lattice point counting which can be used to obtain the same result we got here. For a single-variable problem ($\ell_i=\ell_j$) one only needs to know the polynomial order and the first few solutions to find the polynomial for any $\ell$ by using, for example, the Lagrange method. For multivariate problems, such as the original system~\eqref{eq:Diophantine}, one can use Barvinok's algorithm~\cite{barvinok1999algorithmic} or the approach by MacMahon called partition analysis~\cite{macmahon1916combinatory} originally developed for other purposes. These could be called `black box' methods\footnote{There exist SW packages such as~\href{http://www.math.ucdavis.edu/~latte}{LattE} or~\href{http://www.risc.jku.at/research/combinat/software/ergosum/RISC/Omega.html}{Omega} finding the number of solutions instantaneously.} and are not the methods used here. The author's hope is that for the physically relevant problem of many variables $\ell_i\neq\ell_j$ (and for a larger, but similar, system we briefly discuss in Section~\ref{sec:open}) we will be able to use the symmetries of~\eqref{eq:Diophantine} as well as a certain invariant which were instrumental in finding the number of solutions here.

\section{Main result}\label{sec:main}

\begin{theorem}\label{thm:main}
  The number of nonnegative solutions of system~\eqref{eq:Diophantine} is given by
  \begin{equation}\label{eq:ellEven}
    \mathsf{e}(\ell)=\frac{1}{576} (\ell+2) (\ell+4)\big(\ell (\ell+5) (\ell (\ell+4)+12)+72\big)
  \end{equation}
  for $\a_{ij}\in\bbZ_{\geq0}$ and $\ell=\ell_i$ even and
  \begin{equation}\label{eq:ellOdd}
    \mathsf{d}(\ell)=\frac{1}{576} (\ell+1) (\ell+3) \big(\ell (\ell+5) (\ell (\ell+6)+17)+72\big)
  \end{equation}
  for $\ell=\ell_i$ odd.
\end{theorem}
\begin{remark}
  By remapping $\ell\mapsto2\ell-1$ in~\eqref{eq:ellEven} and $\ell\mapsto2\ell-3$ in~\eqref{eq:ellOdd} we get
    \begin{equation}\label{eq:ellEvenRemap}
    \mathsf{\tilde e}(\ell)={1\over18} \ell (\ell+1) (3+2 \ell+\ell^2+\ell^3+2 \ell^4)
  \end{equation}
  for $\ell\geq1$ and
  \begin{equation}\label{eq:ellOddRemap}
    \mathsf{\tilde d}(\ell)={1\over18} \ell(\ell-1)(3-2 \ell+\ell^2-\ell^3+2 \ell^4)
  \end{equation}
  for $\ell\geq2$, showing a certain similarity.
\end{remark}
Let us recall the definition of the square lattice and all other useful concepts we will use here\footnote{The symbol $\df$ used here stands for `define'.}.
\begin{definition}\label{def:lattice}
    The \emph{square lattice} is the set $\bbZ^2\df\{(k,l);k,l\in\bbZ\}$ as a subset of $\bbR^2$. The nonnegative even quadrant is the set $\bbZ_{\mathrm e}^2\df\{(k,l);k=2m,l=2n;m,n\in\bbZ_{\geq0}\}$ and the positive odd quadrant is defined as $\bbZ_{\mathrm d}^2\df\{(k,l);k=2m+1,l=2n+1;m,n\in\bbZ_{\geq0}\}$. The elements of $\bbZ_{\mathrm e}^2$ or $\bbZ_{\mathrm d}^2$ are referred to as \emph{vertices} and the path connecting two neighboring vertices is called a \emph{segment}. An affine hyperplane is defined as $H\df\{x_i\in\bbR^2;ax_1+bx_2-c=0\}$ where $a,b,c\in\bbR$. A hyperplane is called \emph{reciprocal, horizontal} and \emph{vertical} by setting $c\in\bbZ_{\geq0}$ and $(a,b)=(1,1)$, $(a,b)=(0,1)$ and $(a,b)=(1,0)$ (in this order).
\end{definition}
\begin{remark}
  The length of any segment in the even and odd square lattice is two. This is the reason for a frequent occurrence of the factor of one half in the upcoming lemmas where we count the number of vertices.
\end{remark}
\begin{lemma}\label{lem:DioSimSystem}
  Considering  $\ell_{i}=\ell,\forall i$ in~\eqref{eq:Diophantine}, let $\ell_{ii}=\ell-2\a_{ii}\geq0$ and $\D\df-\ell_{11}-\ell_{22}+\ell_{33}+\ell_{44}$. Then, for  $\ell_{11}\leq\ell_{22}$ and  $\D\leq0$, there exists a nonnegative solution  for any $\ell_{33}$ and $\ell_{44}$ satisfying
  \begin{equation}\label{eq:DeltaBound}
    -\ell_{11}+\ell_{22}\leq\ell_{33}+\ell_{44}.
  \end{equation}
  Furthermore, $\D$ classifies all nonnegative solutions  according to whether $\D\lessgtr0$ or $\D=0$ and the number of nonnegative solutions for $\D>0$ equals the number od solutions for $\D<0$. Finally,  any  pair $(\a_{12},\a_{34})$ consistent with $\ell_{ii}$ satisfying~\eqref{eq:DeltaBound} determines the total number of nonnegative  solutions $(\a_{ij})_{1\leq i<j\leq4}$ calculated from the following expression:
  \begin{equation}\label{eq:UpperBounda34}
    \min{[\ell_{11}+\D/2,\min{[\ell_{33},\ell_{44}]}]}.
  \end{equation}
\end{lemma}
\begin{proof}
  We rewrite~\eqref{eq:Diophantine} as
  \begin{subequations}\label{eq:DiophantineSim}
      \begin{align}
        \a_{12}+\a_{13}+\a_{14} & = \ell_{11},\label{eq:Diophantine1Sim} \\
        \a_{12}+\a_{23}+\a_{24} & = \ell_{22},\label{eq:Diophantine2Sim} \\
        \a_{13}+\a_{23}+\a_{34} & = \ell_{33},\label{eq:Diophantine3Sim} \\
        \a_{14}+\a_{24}+\a_{34} & = \ell_{44}.\label{eq:Diophantine4Sim}
      \end{align}
    \end{subequations}
  and add~\eqref{eq:Diophantine1Sim} and~\eqref{eq:Diophantine2Sim} followed by subtraction from the sum of the last two lines of~\eqref{eq:DiophantineSim}. We get
  \begin{equation}\label{eq:DeltaDef}
    \D=-\ell_{11}-\ell_{22}+\ell_{33}+\ell_{44}=2\a_{34}-2\a_{12}.
  \end{equation}
  We are looking for nonnegative solutions and so the lower bound $\a_{34}\geq0$ holds. Then, from~\eqref{eq:DeltaDef} we get  $\min{[\ell_{11},\ell_{22}]}+\D/2=\ell_{11}+\D/2\geq0$ that becomes~\eqref{eq:DeltaBound}. We also see that a solution of~\eqref{eq:DeltaDef} exists for any $\D\leq0$ by rewriting~\eqref{eq:DeltaBound} as $ -\ell_{11}+\ell_{22}+c=\ell_{33}+\ell_{44}$ for $c\in\bbZ_{\geq0}$, inserting the RHS to the middle expression in~\eqref{eq:DeltaDef} and setting $\a_{34}=0$. We get $-2\ell_{11}+c=-2\a_{12}$. Since $0\leq\a_{12}\leq\min{[\ell_{11},\ell_{22}]}=\ell_{11}$ the worst-case scenario is $c=0$ and even in that case the equation can be satisfied by setting $\a_{12}=\ell_{11}$. We also have an upper bound  $\a_{34}\leq\min{[\ell_{33},\ell_{44}]}$ but there is no guarantee that $\a_{34}$ can take on all the values. From~\eqref{eq:DeltaDef} it follows that any such value must be `matched' by $\min{[\ell_{11},\ell_{22}]}+{\D/2}$. Hence,  we choose from the two competing quantities $\min{[\min{[\ell_{11},\ell_{22}]}+{\D/2},\min{[\ell_{33},\ell_{44}]}]}$ and considering  $\ell_{11}\leq\ell_{22}$ we arrive at $\a_{34}$~to be upper bounded by~\eqref{eq:UpperBounda34}.  When is~expression~\eqref{eq:UpperBounda34} minimized by the first argument? The question is when $\ell_{11}+\D/2<\min{[\ell_{33},\ell_{44}]}$ holds. Taking into account both possibilities, $\ell_{33}\leq\ell_{44}$ and $\ell_{33}>\ell_{44}$, we obtain the inequality
  \begin{equation}\label{eq:Strip}
    \ell_{22}-\ell_{11}>|\ell_{33}-\ell_{44}|.
  \end{equation}
  Eq.~\eqref{eq:Strip} contains an important piece of information. By searching for all nonnegative solutions we are after all possible nonnegative six-tuples $(\a_{ij})_{1\leq i<j\leq4}$. Naturally, many of them contain the same pair $(\a_{12},\a_{34})$ and so we have to find the pairs' multiplicities to count all the solutions. Due to Eq.~\eqref{eq:DeltaDef} the multiplicity of $\a_{12}$ equals the multiplicity of~$\a_{34}$ which is determined by the value of $\a_{34}$ itself. Eq.~\eqref{eq:UpperBounda34} provides the greatest value $\a_{34}$ can achieve and~\eqref{eq:Strip} tells us where the two possibilities happen. But $\a_{34}$ is not the multiplicity itself. For an admissible $\a_{34}$ there is $\ell_{33}-\a_{34}+1$ of ways $\a_{13}+\a_{23}$ sums to $\ell_{33}-\a_{34}$ in~\eqref{eq:Diophantine3Sim} or equivalently in~\eqref{eq:Diophantine4Sim}. Summing over all allowed $\a_{34}$ we find the multiplicity factor to be a triangle number~\seqnum{A000217}.

  So far we considered separately $\D=0$ and $\D<0$. The latter is equivalent to
  \begin{equation}\label{eq:DeltaNeg}
    \ell_{11}+\ell_{22}>\ell_{33}+\ell_{44}
  \end{equation}
  and we can indeed afford to consider only these two cases. This is because linear system~\eqref{eq:DiophantineSim} is invariant w.r.t. the relabeling $1\leftrightharpoons3$ and $2\leftrightharpoons4$ and the permutation flips the sign of $\D$.
\end{proof}
\begin{remark}
   We will find the explicit expressions for the number of solutions for $\D\leq0$  in  Lemma~\ref{lem:countingEVENell} and~\ref{lem:countingODDell}. It is convenient to depict the found inequalities in a nonnegative quadrant of a square lattice whose segment has length two as introduced in Definition~\ref{def:lattice}. The quadrant's axes are identified with $\ell_{33}$ and $\ell_{44}$ and inequalities~\eqref{eq:DeltaNeg} and~\eqref{eq:DeltaBound} together with the upper bound on $\ell_{33}$ and $\ell_{44}$ demarcate a polygon whose boundary and interior contain all admissible pairs $(\ell_{33},\ell_{44})$ leading to the solutions of~\eqref{eq:DiophantineSim}. An area given by inequalities~\eqref{eq:DeltaBound},\eqref{eq:Strip} and~\eqref{eq:DeltaNeg} will be referred to as a diagonal \emph{strip} and it further splits the polygon into several regions. Different rules for calculating the multiplicities hold in different parts of the polygon and a special care will be taken for the degenerate strip when $\ell_{11}=\ell_{22}$.
\end{remark}
\begin{lemma}\label{lem:countingEVENell}
    Given the assumptions of Lemma~\ref{lem:DioSimSystem} and for $\ell$ even the number of nonnegative solutions of~Eqs.~\eqref{eq:DiophantineSim} is the sum of the following expressions:
    \begin{subequations}\label{eq:NoOfSolsEven}
    \begin{align}
      \mathsf{E}^{1,\D_-}_{<} & = \sum_{t=1}^{\ell_{11}/2}{2t(2t-1)\over2}\bigg({\ell_{22}-\ell_{11}\over2}-1\bigg)+
      \sum_{t=1}^{\ell_{11}/2}{(2t+1)2t\over2}{\ell_{22}-\ell_{11}\over2}\nn\\
      &\quad+2\sum_{t=1}^{\ell_{11}/2}{2t(2t-1)\over2}(\ell_{11}-(2t-2)),      \label{eq:NoOfSolsA}\\
      \msE^{1,\D_-}_{=} & = -\sum_{t=1}^{\ell_{11}/2}{2t(2t-1)\over2}+2\sum_{t=1}^{\ell_{11}/2}{2t(2t-1)\over2}(\ell_{11}-(2t-2)),\label{eq:NoOfSolsB}\\
      \msE^{2,\D_-}_{<} & = \sum_{t=1}^{\ell_{11}/2}{2t(2t-1)\over2}\bigg({\ell_{22}-\ell_{11}\over2}-1\bigg)+
      \sum_{t=1}^{\ell_{11}/2}{(2t+1)2t\over2}{\ell_{22}-\ell_{11}\over2}\nn\\
      &\quad+2\sum_{t=1}^{\ell_{11}/2}{2t(2t-1)\over2}(\ell_{11}-(2t-2))
          -2\sum_{t=1}^{(\ell_{11}+\ell_{22}-\ell)/2-1}{2t(2t-1)\over2}\bigg({\ell_{11}+\ell_{22}-\ell\over2}-t\bigg),\label{eq:NoOfSolsC}\\
      \msE^{2,\D_-}_{=} & = -\sum_{t=1}^{\ell_{11}/2}{2t(2t-1)\over2}\nn\\
          &\quad+2\sum_{t=1}^{\ell_{11}/2}{2t(2t-1)\over2}(\ell_{11}-(2t-2))
          -2\sum_{t=1}^{(2\ell_{11}-\ell)/2-1}{2t(2t-1)\over2}\bigg({2\ell_{11}-\ell\over2}-t\bigg),\label{eq:NoOfSolsD}\\
      \msE^{1,\D_0}_{<} & = {1\over2}(\ell_{11}+1)(\ell_{11}+2)\big((\ell_{22}-\ell_{11})/2-1\big)+2\sum_{t=1}^{\ell_{11}/2+1}{2t(2t-1)\over2},\label{eq:NoOfSolsE}\\
      \msE^{1,\D_0}_{=} & = -{1\over2}(\ell_{11}+1)(\ell_{11}+2)+2\sum_{t=1}^{\ell_{11}/2+1}{2t(2t-1)\over2},\label{eq:NoOfSolsF}\\
      \msE^{2,\D_0}_{<} & ={1\over2}(\ell_{11}+1)(\ell_{11}+2)\big((\ell_{22}-\ell_{11})/2-1\big)\nn\\
           &\quad+2\sum_{t=1}^{(\ell-\ell_{22})/2+1}{1\over2}(\ell_{11}+\ell_{22}-\ell+2t-1)(\ell_{11}+\ell_{22}-\ell+2t), \label{eq:NoOfSolsG}\\
      \msE^{2,\D_0}_{=} & = -{1\over2}(\ell_{11}+1)(\ell_{11}+2)
    +2\sum_{t=1}^{(\ell-\ell_{11})/2+1}{1\over2}(2\ell_{11}-\ell+2t-1)(2\ell_{11}-\ell+2t),\label{eq:NoOfSolsH}
    \end{align}
    \end{subequations}
    where $\D_-,\D_0$ denote $\D<0$ and $\D=0$, respectively, and the subscripts $<$ and $=$ distinguish between   $\ell_{11}<\ell_{22}$ and $\ell_{11}=\ell_{22}$.
\end{lemma}
\begin{remark}
  The split into eight cases will become relevant in the proof of Theorem~\ref{thm:main}. For the same reason, there is no need to evaluate the sums at this point.
\end{remark}
\begin{proof}
    For $\D<0$ it is advantageous to distinguish between the following two cases: $\ell\geq\ell_{11}+\ell_{22}-2$ and $\ell_{11}+\ell_{22}-2>\ell$. The first inequality combined with~\eqref{eq:DeltaNeg} implies $\ell\geq\ell_{33}+\ell_{44}$. Since neither of $\ell_{33},\ell_{44}$ can be greater than $\ell$ it follows that two polygon vertices lie on the quadrant axes (connected by the line $\ell=\ell_{33}+\ell_{44}$). The same holds for $\D=0\Leftrightarrow\ell_{11}+\ell_{22}=\ell_{33}+\ell_{44}$ where we separately investigate $\ell\geq\ell_{11}+\ell_{22}$ and $\ell_{11}+\ell_{22}>\ell$.

  \subsection*{Case $\D_-$ and $\ell\geq\ell_{11}+\ell_{22}-2$}
  Let us consider $\ell_{11}<\ell_{22}$ first. To count the points in the strip we will use the reciprocal hyperplanes\footnote{We cannot use Pick's theorem~\cite{beck2007computing}  as different points have different multiplicities we have to take into account.} introduced in Definition~\ref{def:lattice}. All points in the even square lattice lie on the reciprocal hyperplanes delimited by~\eqref{eq:DeltaBound} and~\eqref{eq:DeltaNeg} which gives us a very convenient way of labeling \emph{and} counting of the hyperplanes: $0\leq\ell_{11}+\D/2\leq\ell_{11}-1$.  Inequalities~\eqref{eq:DeltaBound} and~\eqref{eq:DeltaNeg} are saturated when $-\ell_{11}+\ell_{22}=\ell_{33}+\ell_{44}$ and $\ell_{11}+\ell_{22}-2=\ell_{33}+\ell_{44}$, respectively. It follows that there is $2\ell_{11}-2$ segments between the intersection points of these two lines with the axis $\ell_{33}$ or $\ell_{44}$. It also means that there is $(2\ell_{11}-2)/2+1=\ell_{11}$ reciprocal hyperplanes. There are two types of reciprocal hyperplanes. One type intersects $(\ell_{22}-\ell_{11}-2)/2+1=(\ell_{22}-\ell_{11})/2$ points and the other passes through   $(\ell_{22}-\ell_{11})/2-1$ points. This can be seen in the following way. The diagonal strip boundaries intersect $\ell_{11}+\ell_{22}-2=\ell_{33}+\ell_{44}$ at two points whose $\ell_{33}$ coordinates are $\ell_{11}-1$ and $\ell_{22}-1$. Their distance (projected onto the $\ell_{33}$ or $\ell_{44}$ axis) is $\ell_{22}-\ell_{11}$ but because these are odd coordinates no solution can lie on any vertical or horizontal line passing through them. The closest `even' points inside the strip are one segment away (from each `odd' point) and that is how we got the $(\ell_{22}-\ell_{11})/2$ points above. Thus, the neighboring reciprocal hyperplane passes through $(\ell_{22}-\ell_{11})/2-1$ points. Since we counted the number of hyperplanes to be $\ell_{11}$ (which is even) there is $\ell_{11}/2$ of them for both types. Hence the strip contains
    \begin{align}\label{eq:stripSolsIlessJ}
      s^{1,\D_-}_{<}&=\sum_{t=1}^{\ell_{11}/2}{2t(2t-1)\over2}\bigg({\ell_{22}-\ell_{11}\over2}-1\bigg)+
      \sum_{t=1}^{\ell_{11}/2}{(2t+1)2t\over2}{\ell_{22}-\ell_{11}\over2}
    \end{align}
  solutions. The $t$ parameter is set up such that it takes the corresponding values from the interval $0\leq\ell_{11}+\D/2\leq\ell_{11}-1$ governing the multiplicity factor.

  We will use the vertical hyperplanes to count the number of solutions for the rest of the polygon. In the subset where $\ell_{33}<\ell_{44}$ holds it is (see Eq.~\eqref{eq:UpperBounda34}) $\ell_{33}$ according to which the multiplicities  are calculated.
  The upper diagonal strip boundary intersects the $\ell_{44}$ axis at $\ell_{22}-\ell_{11}$ and the line $\ell_{11}+\ell_{22}-2=\ell_{33}+\ell_{44}$ intersects the axis at $\ell_{11}+\ell_{22}-2$. So there is $\ell_{11}$ vertices with nonnegative solutions. Every time  $\ell_{33}$ increases by two we get two points less and from the previous paragraph the maximal value of $\ell_{33}$ is $\ell_{11}-1-1=\ell_{11}-2$. Hence there is $(\ell_{11}-2-0)/2+1=\ell_{11}/2$ vertical axes. For $\ell_{33}>\ell_{44}$  the situation is verbatim where instead of vertical hyperplanes we study horizontal hyperplanes in the mirror image across the diagonal. Hence, the number of solutions reads
    \begin{equation}\label{eq:outsideStripSolsIlessJ}
      r^{1,\D_-}_{<}=2\sum_{t=1}^{\ell_{11}/2}{2t(2t-1)\over2}(\ell_{11}-(2t-2)).
    \end{equation}
  Summing Eqs.~\eqref{eq:stripSolsIlessJ} and~\eqref{eq:outsideStripSolsIlessJ} we obtain~\eqref{eq:NoOfSolsA}.

  For  $\ell_{11}=\ell_{22}$ the strip becomes a diagonal line. The counting with the help of vertical and horizontal line goes through in exactly the same way leading to~Eq.~\eqref{eq:outsideStripSolsIlessJ}. The diagonal solutions are, however, doubly counted since the strip is degenerate and must be subtracted. This is precisely the first term of Eq.~\eqref{eq:stripSolsIlessJ}. Therefore, for the number of solutions we get~\eqref{eq:NoOfSolsB}.

  \subsection*{Case $\D_-$ and $\ell<\ell_{11}+\ell_{22}-2$}
  Starting with $\ell_{11}<\ell_{22}$ and $\ell_{33}<\ell_{44}$, the sum $\ell_{33}+\ell_{44}$ is bounded only by~\eqref{eq:DeltaNeg} together with $\ell_{33},\ell_{44}\leq\ell$. So we insert $\ell_{33}=0$ and $\ell_{44}=\ell$ to~\eqref{eq:DeltaNeg} and then the expression $(\ell_{11}+\ell_{22}-\ell-2)/2$ counts the number of horizontal steps from the polygon vertex point $(0,\ell)$. Hence, the polygon's shape is now more complicated -- there are two more vertices on the line given by $\ell_{11}+\ell_{22}-2=\ell_{33}+\ell_{44}$. 
  It is advantageous to let the vertical hyperplanes (recall that $\ell_{33}<\ell_{44}$ is being considered) count until they hit $\ell_{11}+\ell_{22}-2=\ell_{33}+\ell_{44}$ and then subtract the inadmissible solutions -- those above the `cut-off' line $\ell_{44}=const$. The cut-off line is always $\ell_{44}=\ell$ since $\ell_{44}$ can reach it but cannot go higher ($\ell_{33},\ell_{44}\leq\ell$)\footnote{Note that the point $(\ell_{33},\ell_{44})=(0,\ell)$ satisfies constraint~\eqref{eq:DeltaNeg} unless $\ell_{11}=0$ which, however, corresponds to $\D=0$ solved as a separate case.}.
  Hence, for the number of solutions we get
        \begin{equation}\label{eq:outsideStripSolsIlessJcaseII}
          r^{2,\D_-}_{<}=2\sum_{t=1}^{\ell_{11}/2}{2t(2t-1)\over2}(\ell_{11}-(2t-2))
          -2\sum_{t=1}^{(\ell_{11}+\ell_{22}-\ell)/2-1}{2t(2t-1)\over2}\bigg({\ell_{11}+\ell_{22}-\ell\over2}-t\bigg),
        \end{equation}
  where the first term is identical to~\eqref{eq:outsideStripSolsIlessJ} and the upper bound in the second sum is given by counting the inadmissible solutions: we set $\ell_{33}=0$ in $\ell_{11}+\ell_{22}-2=\ell_{33}+\ell_{44}$, find $\ell_{44}$ and calculate $\ell_{44}-\ell=\ell_{11}+\ell_{22}-\ell-2$. So the number of vertices on the $\ell_{44}$ axis is $(\ell_{11}+\ell_{22}-\ell-2)/2+1-1$ leading to the sum's upper bound. It is also the expression in the parenthesis where the $t$ variable is set up such that vertical hyperplanes with the decreasing number of solutions (by one) are assigned the correct multiplicity factors (in the form of the triangle numbers) as revealed in Lemma~\ref{lem:DioSimSystem}. The factor of two again accounts for the solutions from mirror image situation  on the other side of the strip for $\ell_{44}<\ell_{33}$ (using  horizontal hyperplanes).

   Counting in the strip is the same as in~\eqref{eq:stripSolsIlessJ}. This is because the cut-off line never violates the points inside the strip. The cut-off line  $\ell_{44}=\ell$  intersects $\ell_{11}+\ell_{22}-2=\ell_{33}+\ell_{44}$ at $\ell_{33}=\ell_{11}-2$. By inserting this value to the upper diagonal strip boundary $\ell_{22}-\ell_{11}=-\ell_{33}+\ell_{44}$ we can see that the cut-off line cannot even get to the strip boundary. For $\ell_{44}<\ell_{33}$ we arrive at the same conclusion and so $s^{2,\D_-}_{<}=s^{1,\D_-}_{<}$ from~\eqref{eq:stripSolsIlessJ} and together with~\eqref{eq:outsideStripSolsIlessJcaseII} we get~\eqref{eq:NoOfSolsC}.

  The case $\ell_{11}=\ell_{22}$ is again a special instance of~\eqref{eq:NoOfSolsC} thus reducing it to~\eqref{eq:NoOfSolsD}.

  \subsection*{Case $\D_0$ and $\ell\geq\ell_{11}+\ell_{22}$}
  Let us recall that $\D=0$  translates into
    \begin{equation}\label{eq:DeltaZero}
    \ell_{11}+\ell_{22}=\ell_{33}+\ell_{44}.
    \end{equation}
  So now it is advantageous to separately investigate $\ell\geq\ell_{11}+\ell_{22}$ and $\ell<\ell_{11}+\ell_{22}$. Similarly to the $\D<0$ case, the first inequality implies $\ell\geq\ell_{33}+\ell_{44}$ with the same consequences for the polygon vertices. Contrary to $\D<0$ we will use the reciprocal hyperplanes (just one to be precise) to count the solutions. This is because now all the solutions  lie on~\eqref{eq:DeltaZero}.  For $\ell_{11}<\ell_{22}$ the strip defined by~\eqref{eq:Strip} becomes a line containing $(\ell_{22}-\ell_{11})/2-1$ solutions. We derived the number governing their multiplicity (see right before~\eqref{eq:Strip}) to be $\ell_{11}+\D/2=\ell_{11}$ and so the strip contributes with
  \begin{equation}\label{eq:stripDzerosolsiLj}
    s^{1,\D_0}_{<}={1\over2}(\ell_{11}+1)(\ell_{11}+2)\big((\ell_{22}-\ell_{11})/2-1\big)
  \end{equation}
  solutions. Since $\ell_{22}-\ell_{11}=-\ell_{33}+\ell_{44}$ intersect at $(\ell_{11},\ell_{22})$ there is $\ell_{11}/2+1$ points (lying on~\eqref{eq:DeltaZero}) between the strip boundary and the $\ell_{44}$ axis. The multiplicity is now based on~$\ell_{33}$ and taking into account doubling from the same argument for $\ell_{44}<\ell_{33}$ the number of solutions lying on~\eqref{eq:DeltaZero} reads
  \begin{equation}\label{eq:OutsidestripDzerosolsiLj}
    r^{1,\D_0}_{<}=2\sum_{t=1}^{\ell_{11}/2+1}{2t(2t-1)\over2}.
  \end{equation}
  Summing~\eqref{eq:stripDzerosolsiLj} and~\eqref{eq:OutsidestripDzerosolsiLj} we get~\eqref{eq:NoOfSolsE}.

  For  $\ell_{11}=\ell_{22}$ the strip intersections with~\eqref{eq:DeltaZero} is just a point and we simply add~\eqref{eq:stripDzerosolsiLj} and~\eqref{eq:OutsidestripDzerosolsiLj} resulting in~\eqref{eq:NoOfSolsF}.
  The negative contribution removes the overlapping point shared by the $\ell_{44}<\ell_{33}$ and $\ell_{44}>\ell_{33}$ solutions.

  \subsection*{Case $\D_0$ and $\ell<\ell_{11}+\ell_{22}$}
  The presence of a cut-off line $\ell_{44}=const$ has again no effect on the intersection of the strip and~\eqref{eq:DeltaZero}. As before, the lowest cut-off is $\ell_{44}=\ell_{22}=\ell$ and it intersects~\eqref{eq:DeltaZero} at $\ell_{33}=\ell_{11}$, that is, precisely at the intersection boundary given by $\ell_{22}-\ell_{11}=-\ell_{33}+\ell_{44}$. Hence the number of solutions is as in~\eqref{eq:stripDzerosolsiLj} and we write $s^{2,\D_0}_{<}=s^{1,\D_0}_{<}$.  For a generic $\ell_{44}=\ell$ we find that the boundary intersects~\eqref{eq:DeltaZero} at $\ell_{44}=\ell_{22}$ and so there is $(\ell-\ell_{22})/2+1$ points. The multiplicity is governed by $\ell_{33}$ and for $\ell_{44}=\ell$ we get from~\eqref{eq:DeltaZero} $\ell_{33}=\ell_{11}+\ell_{22}-\ell$. As we approach the strip, $\ell_{33}$ increases by two with each lattice segment. Hence, considering the identical calculation for $\ell_{33}>\ell_{44}$, we get
  \begin{equation}\label{eq:OutsidestripDzerosolsiEQj}
    r^{2,\D_0}_{<}=2\sum_{t=1}^{(\ell-\ell_{22})/2+1}{1\over2}(\ell_{11}+\ell_{22}-\ell+2t-1)(\ell_{11}+\ell_{22}-\ell+2t)
  \end{equation}
  and so~\eqref{eq:NoOfSolsG} follows. The case  $\ell_{11}=\ell_{22}$ follows as in~\eqref{eq:NoOfSolsH}.
\end{proof}
\begin{lemma}\label{lem:countingODDell}
    Given the assumptions of Lemma~\ref{lem:DioSimSystem} and for $\ell$ odd the number of nonnegative solutions of~Eq.~\eqref{eq:DiophantineSim} is a sum of the following expressions:
    \begin{subequations}\label{eq:NoOfSolsOdd}
    \begin{align}
      \msD^{1,\D_-}_{<} & =    \sum_{t=1}^{(\ell_{11}-1)/2}{2t(2t+1)\over2}\bigg({\ell_{22}-\ell_{11}\over2}-1\bigg)
                            +\sum_{t=1}^{(\ell_{11}+1)/2}{(2t-1)2t\over2}{\ell_{22}-\ell_{11}\over2}\nn\\
                     &\quad+2\sum_{t=1}^{(\ell_{11}-1)/2}{2t(2t+1)\over2}(\ell_{11}-(2t-2)-1),  \label{eq:NoOfSolsAodd}\\
      \msD^{1,\D_-}_{=} & =  -\sum_{t=1}^{(\ell_{11}-1)/2}{2t(2t+1)\over2}+2\sum_{t=1}^{(\ell_{11}-1)/2}{2t(2t+1)\over2}(\ell_{11}-(2t-2)-1),
                                        \label{eq:NoOfSolsBodd}\\
      \msD^{2,\D_-}_{<} & =    \sum_{t=1}^{(\ell_{11}-1)/2}{2t(2t+1)\over2}\bigg({\ell_{22}-\ell_{11}\over2}-1\bigg)+
                            \sum_{t=1}^{(\ell_{11}+1)/2}{(2t-1)2t\over2}{\ell_{22}-\ell_{11}\over2}\nn\\
                     &\quad+2\sum_{t=1}^{(\ell_{11}-1)/2}{2t(2t+1)\over2}(\ell_{11}-(2t-2)-1)\nn\\
                     &\quad-2\sum_{t=1}^{(\ell_{11}+\ell_{22}-\ell-3)/2}{2t(2t+1)\over2}\bigg({\ell_{11}+\ell_{22}-\ell-1\over2}-t\bigg),\label{eq:NoOfSolsCodd}\\
      \msD^{2,\D_-}_{=} & =   -\sum_{t=1}^{(\ell_{11}-1)/2}{2t(2t+1)\over2}+2\sum_{t=1}^{(\ell_{11}-1)/2}{2t(2t+1)\over2}(\ell_{11}-(2t-2)-1)\nn\\
                     &\quad-2\sum_{t=1}^{(2\ell_{11}-\ell-3)/2}{2t(2t+1)\over2}\bigg({2\ell_{11}-\ell-1\over2}-t\bigg)\label{eq:NoOfSolsDodd}\\
      \msD^{1,\D_0}_{<} & =   {1\over2}(\ell_{11}+1)(\ell_{11}+2)\big((\ell_{22}-\ell_{11})/2-1\big)+2\sum_{t=1}^{(\ell_{11}+1)/2}{2t(2t+1)\over2},
                                    \label{eq:NoOfSolsEodd}\\
      \msD^{1,\D_0}_{=} & =  - {1\over2}(\ell_{11}+1)(\ell_{11}+2)+2\sum_{t=1}^{(\ell_{11}+1)/2}{2t(2t+1)\over2},
                                    \label{eq:NoOfSolsFodd}\\
      \msD^{2,\D_0}_{<} & =   {1\over2}(\ell_{11}+1)(\ell_{11}+2)\big((\ell_{22}-\ell_{11})/2-1\big)\nn\\
                     &\quad+2\sum_{t=1}^{(\ell-\ell_{22})/2+1}{1\over2}(\ell_{11}+\ell_{22}-\ell+2t-1)(\ell_{11}+\ell_{22}-\ell+2t),
                                    \label{eq:NoOfSolsGodd}\\
      \msD^{2,\D_0}_{=} & =  -{1\over2}(\ell_{11}+1)(\ell_{11}+2)+2\sum_{t=1}^{(\ell-\ell_{11})/2+1}{1\over2}(2\ell_{11}-\ell+2t-1)(2\ell_{11}-\ell+2t),
                                    \label{eq:NoOfSolsHodd}
    \end{align}
    \end{subequations}
    where $\D_-,\D_0$ denote $\D<0$ and $\D=0$, respectively, and the subscripts $<$ and $=$ distinguish between   $\ell_{11}<\ell_{22}$ and $\ell_{11}=\ell_{22}$.
\end{lemma}
\begin{remark}
   There does not seem to exist an easy way of applying the even $\ell$ results to the odd case. The proof, however, bears similarities to the proof of Lemma~\ref{lem:countingEVENell} including the split into several (eight) cases. This makes  counting easier and also serves for the sake of proof of Theorem~\ref{thm:main}. One of the cases we need to consider separately is when $\ell_{11}=\ell_{22}$. It turns out to be given by the $\ell_{11}<\ell_{22}$ case (by setting $\ell_{11}=\ell_{22}$) like in Lemma~\ref{lem:countingEVENell}.
\end{remark}
\begin{proof}
    Here it is advantageous to distinguish between  $\ell\geq\ell_{11}+\ell_{22}-3$ and $\ell_{11}+\ell_{22}-3>\ell$. The first inequality combined with~\eqref{eq:DeltaNeg} implies $\ell+1\geq\ell_{33}+\ell_{44}$. Since neither of $\ell_{33},\ell_{44}$ can be greater than $\ell$ it follows that two polygon vertices lie on the quadrant axes (connected by the line $\ell=\ell_{33}+\ell_{44}$). When $\D=0$  we separately investigate $\ell\geq\ell_{11}+\ell_{22}-1$ and $\ell<\ell_{11}+\ell_{22}-1$ for the same reason.

    An important difference compared to Lemma~\ref{lem:countingEVENell} is the location of the positive axes $\ell_{33}$ and $\ell_{44}$ in the odd square lattice introduced in Definition~\ref{def:lattice}. The axis $\ell_{33}$ will be identified with $x_2=1$ and $\ell_{44}$ with $x_1=1$. The reason is that unlike the even case, the solution-counting vertices in the square lattice lie on the odd coordinates and the smallest odd number is one. Because there is no nonnegative solution lying on a $(0,i)$ or $(j,0)$ we will use the shifted coordinate system in the next four subsections.

    \subsection*{Case $\D_-$ and $\ell\geq\ell_{11}+\ell_{22}-3$}
    Consider $\ell_{11}<\ell_{22}$. We will use the reciprocal hyperplanes to count the solutions in the strip area as they are characterized by $0\leq\ell_{11}+\D/2\leq\ell_{11}-1$ shown in Lemma~\ref{lem:DioSimSystem}. Exactly as in the proof of Lemma~\ref{lem:countingEVENell} (Case $\D<0$ and $\ell\geq\ell_{11}+\ell_{22}-2$) we find the diagonal strip intersection to be at two points whose $\ell_{33}$ coordinates are $\ell_{11}-1$ and $\ell_{22}-1$. Again, their distance (projected onto the $\ell_{33}$ or $\ell_{44}$ axis) is $\ell_{22}-\ell_{11}$ and here the analysis starts to differ. The coordinates $\ell_{11}-1$ and $\ell_{22}-1$ are even so no solution can lie on any vertical or horizontal line intersecting them. The closest `odd' points inside the strip are one segment away (from each `even' point) and so there is $(\ell_{22}-\ell_{11})/2$ vertices. Consequently, the neighboring reciprocal hyperplane intersects $(\ell_{22}-\ell_{11})/2-1$ vertices. Since the first and last hyperplane (given by $-\ell_{11}+\ell_{22}=\ell_{33}+\ell_{44}$ and $\ell_{11}+\ell_{22}-2=\ell_{33}+\ell_{44}$, respectively) counts $(\ell_{22}-\ell_{11})/2$ solutions and the total number  of hyperplanes is $\ell_{11}$ (odd) we get the first two summands of~\eqref{eq:NoOfSolsAodd} with different upper bounds.

    The vertical hyperplanes will be used for the region outside the strip where $\ell_{33}<\ell_{44}$ since the multiplicity factor is given by  their $\ell_{33}$ coordinate. The upper diagonal strip boundary intersects the $\ell_{44}$ axis at $\ell_{22}-\ell_{11}+1$ and $\ell_{11}+\ell_{22}-2=\ell_{33}+\ell_{33}$ intersects it at $\ell_{11}+\ell_{22}-3$. So there is $\ell_{11}-1$ vertices with nonnegative solutions. Every time  $\ell_{33}$ increases by two we get two points less and from the previous paragraph the maximal value of $\ell_{33}$ is $\ell_{11}-1-1=\ell_{11}-2$. Hence there is $(\ell_{11}-3)/2+1=(\ell_{11}-1)/2$ vertical axes and the last summand of~\eqref{eq:NoOfSolsAodd} is found (multiplied by two to account for the mirror case $\ell_{33}>\ell_{44}$).

    Eq.~\eqref{eq:NoOfSolsBodd} is obtained by setting $\ell_{11}=\ell_{22}$.

    \subsection*{Case $\D_-$ and $\ell<\ell_{11}+\ell_{22}-3$}
    The situation is very similar to the relation between~\eqref{eq:NoOfSolsA} and~\eqref{eq:NoOfSolsC} so we only highlight a different step. The first three terms of~\eqref{eq:NoOfSolsCodd} are the same as in~\eqref{eq:NoOfSolsAodd} and the last term removes the inadmissible solutions above the cut-off line(s) (for $\ell_{33}<\ell_{44}$ and its diagonal mirror image $\ell_{33}>\ell_{44}$). Considering $\ell_{33}<\ell_{44}$, the line $\ell_{11}+\ell_{22}-2=\ell_{33}+\ell_{33}$ intersects the $\ell_{44}$ axis at $\ell_{11}+\ell_{33}-3$ and the distance from the cut-off line $\ell_{44}=\ell$ is $\ell_{11}+\ell_{22}-\ell-3$. So the number of inadmissible vertices on the axis is $(\ell_{11}+\ell_{22}-\ell-3)/2+1-1$. This is the upper bound in the last sum of~\eqref{eq:NoOfSolsCodd} and the expression in the parenthesis. The parameter $t$ is again set up to properly count the inadmissible nonnegative solutions on the vertical hyperplanes together with their multiplicities.

    Eq.~\eqref{eq:NoOfSolsDodd} is obtained by setting $\ell_{22}=\ell_{11}$ in~\eqref{eq:NoOfSolsCodd}.

    \subsection*{Case $\D_0$ and $\ell\geq\ell_{11}+\ell_{22}-1$}
    All solutions lie on the reciprocal hyperplane given by $\D=0\Leftrightarrow\ell_{11}+\ell_{22}=\ell_{33}+\ell_{44}$. The strip solutions lie between the points given by the intersection of~\eqref{eq:Strip} and $\D=0$ whose projection onto the $\ell_{33}$ axis equals $\ell_{11}$ and $\ell_{22}$. So the intersection point are odd and therefore containing admissible nonnegative solutions. Their (projected) distance is $\ell_{22}-\ell_{11}$ and the strip vertices lie between them (on $\D=0$). Henceforth, there is $(\ell_{22}-\ell_{11})/2+1-2$ of them and the multiplicity is calculated from $\min{[\ell_{11},\ell_{22}]}+{\D/2}=\ell_{11}$ according to Lemma~\eqref{lem:DioSimSystem}. This is the first term in~\eqref{eq:NoOfSolsEodd}. For the second term, if $\ell_{33}<\ell_{44}$, there is $\ell_{11}-1$ segments between the strip upper boundary and the $\ell_{44}$ axis and so $(\ell_{11}-1)/2+1$ vertices. The multiplicity is calculated from $\ell_{33}$ and the $t$ parameter in the second term of~\eqref{eq:NoOfSolsEodd} does precisely that. The factor of two accounts for the $\ell_{33}>\ell_{44}$ situation.

    Eq.~\eqref{eq:NoOfSolsFodd} is obtained by setting $\ell_{22}=\ell_{11}$ in~\eqref{eq:NoOfSolsEodd}.

    \subsection*{Case $\D_0$ and $\ell<\ell_{11}+\ell_{22}-1$}

    As before, the cut-off line $\ell_{44}=\ell$ removes some solutions from $\D=0$ (in the $\ell_{33}<\ell_{44}$ case) but always  outside the strip. So the first summand of~\eqref{eq:NoOfSolsGodd} is identical to the first summand of~\eqref{eq:NoOfSolsEodd}. For the part of $\D=0$ outside and on the boundary of the strip we notice that the cut-off line $\ell_{44}=\ell$ intersects $\D=0$ at the point $\ell_{33}=\ell_{11}+\ell_{22}-\ell$ which is $\ell-\ell_{22}$ segments away from the upper diagonal strip boundary point (distance measured by projecting onto the $\ell_{33}$ axis). Hence there is only $(\ell-\ell_{22})/2+1$ admissible vertices on $\D=0$ and we recovered the upper bound of the second sum in~\eqref{eq:NoOfSolsGodd}. We sum over the multiplicity governed by $\ell_{33}$ in this region and that is determined by the $t$ variable in the second sum. As before, for $\ell_{33}>\ell_{44}$ the situation is identical and it brings an overall factor of two.

    The last expression,  Eq.~\eqref{eq:NoOfSolsHodd}, is obtained by setting $\ell_{22}=\ell_{11}$ in~\eqref{eq:NoOfSolsGodd}.

\end{proof}
\begin{lemma}\label{lem:ell1ell22}
  For $\ell_{11}>\ell_{22}$ the number of solutions of Diophantine system~\eqref{eq:DiophantineSim} is equal to the number of solutions for $\ell_{11}<\ell_{22}$ in Lemma~\ref{lem:countingEVENell}. That is
    \begin{subequations}
      \begin{align}
        \msE^{1,\D_-}_{>} & = \msE^{1,\D_-}_{<}, \\
        \msE^{2,\D_-}_{>} & = \msE^{2,\D_-}_{<}, \\
        \msE^{1,\D_0}_{>} & = \msE^{1,\D_0}_{<}, \\
        \msE^{2,\D_0}_{>} & = \msE^{2,\D_0}_{<},
      \end{align}
    \end{subequations}
    and similarly for  odd $\ell$, Eqs.~\eqref{eq:NoOfSolsOdd}, in Lemma~\ref{lem:countingODDell}. The subscript $>$ denotes the case of interest $\ell_{11}>\ell_{22}$.
\end{lemma}
\begin{proof}
   Invariance w.r.t. the permutation $1\leftrightharpoons2$ is another symmetry of~\eqref{eq:DiophantineSim}. The permutation swaps~\eqref{eq:Diophantine1Sim} with~\eqref{eq:Diophantine2Sim} and~\eqref{eq:Diophantine3Sim} with~\eqref{eq:Diophantine4Sim} and keeps $\D$ intact. Hence, if $\ell_{11}>\ell_{22}$ we permute~\eqref{eq:DiophantineSim} and apply Lemma~\ref{lem:DioSimSystem} in order to calculate $\msE^{1,\D_-}_{<}, \msE^{2,\D_-}_{<}, \msE^{1,\D_0}_{<}$ and $\msE^{2,\D_0}_{<}$ in Lemma~\ref{lem:countingEVENell} and $\msD^{1,\D_-}_{<}, \msD^{2,\D_-}_{<}, \msD^{1,\D_0}_{<}$ and $\msD^{2,\D_0}_{<}$ in Lemma~\ref{lem:countingODDell}.
\end{proof}
\begin{proof}[Proof of Theorem~\ref{thm:main}]
    Lemmas~\ref{lem:countingEVENell} and~\ref{lem:countingODDell} counted the solutions for a given $\ell_{11}$ and $\ell_{22}$ so our task is to sum over all such pairs. Because the lemmas are split into several cases we have to adjust the summation procedure accordingly. Basically, all the work is about finding the way to reconcile the condition $\ell_{11}<\ell_{22}$ or $\ell_{11}=\ell_{22}$ with the different investigated cases. There is a technical assumption we have to make. We found~\eqref{eq:ellEven} for $\ell$ even and~\eqref{eq:ellOdd} for $\ell$ odd. However, the summing slightly differs between $4\mid\ell$  and $4\mid\ell-2$ in the former case and $4\mid\ell-1$  and $4\mid\ell-3$ in the latter case. The results are identical and  we present the derivation only for $4\mid\ell$ and $4\mid\ell-1$ in order not to overblow the proof.

    \subsection*{Case $4\mid\ell$ and $\D_-$}
    For $\D<0$ the inequalities $\ell\geq\ell_{11}+\ell_{22}-2$ and $\ell_{11}<\ell_{22}$
    \begin{align}\label{eq:DelNeg_iGj1}
      &\sum_{\ell_{22}=2,4,\dots}^\ell \msE^{1,\D_-}_{>}
      +\sum_{\ell_{11}=2,4,\dots}^{\ell/2}\,\sum_{\ell_{22}=\ell_{11}+2}^{\ell-\ell_{11}+2}\msE^{1,\D_-}_{<}=\frac{1}{46080}\ell (\ell+4) (\ell+8) \left(\ell^3+15 \ell^2+83 \ell+204\right),
    \end{align}
    where in the first sum we set $\ell_{11}=0$ that must be treated separately. In fact, the first summand equals zero. For $\ell<\ell_{11}+\ell_{22}-2$ we get
    \begin{align}\label{eq:DelNeg_iGj2}
      &\sum_{\ell_{11}=\ell/2+2}^{\ell-2}\,\sum_{\ell_{22}=\ell_{11}+2}^{\ell}\msE^{2,\D_-}_{<}
      +\sum_{\ell_{11}=4}^{\ell/2}\,\sum_{\ell_{22}=\ell-\ell_{11}+4}^{\ell}\msE^{2,\D_-}_{<}\nn\\
      &=\frac{1}{46080}(\ell-4) \ell (\ell+4) \left(19 \ell^3+153 \ell^2+509 \ell+528\right).
    \end{align}
    The first term starts counting where the first sum of the second term in~\eqref{eq:DelNeg_iGj1} terminated. The second term in~\eqref{eq:DelNeg_iGj2} starts summing where the second sum of the second term in~\eqref{eq:DelNeg_iGj1} terminated.     By summing~\eqref{eq:DelNeg_iGj1} and~\eqref{eq:DelNeg_iGj2} we obtain
    \begin{equation}\label{eq:DelNeg_iGj}
      \mathsf{e}^{\D_-}_{<}=\frac{1}{2304}\ell (\ell+4) \left(\ell^4+5 \ell^3+5 \ell^2-32 \ell-24\right).
    \end{equation}
    If $\ell_{11}=\ell_{22}$ the counting is simpler. In the first case we have
    \begin{equation}\label{eq:DelNeg_iEj1}
      \sum_{\ell_{11}=2,4,\dots}^{\ell/2}\msE^{1,\D_-}_{=}=\frac{1}{7680}\ell (\ell+2) (\ell+4) (\ell+6) (\ell+8)
    \end{equation}
    (note that $\ell_{11}=\ell_{22}=0$ is excluded since only $\ell_{33}=\ell_{44}$ is admissible and so it belongs to the $\D=0$ case) and in the second case we continue summing by
    \begin{equation}\label{eq:DelNeg_iEJ2}
      \sum_{\ell_{11}=\ell/2+2}^{\ell}\msE^{1,\D_-}_{=}=\frac{1}{7680}\ell (\ell+4) \left(23 \ell^3+148 \ell^2+388 \ell+128\right).
    \end{equation}
    The sum of the last two expressions reads
    \begin{equation}\label{eq:DelNeg_iEJ}
      \mathsf{e}^{\D_-}_{=}=\frac{1}{1920}\ell (\ell+4) \left(6 \ell^3+41 \ell^2+116 \ell+56\right).
    \end{equation}

    \subsection*{Case $4\mid\ell$  and $\D_0$}
    For $\ell\geq\ell_{11}+\ell_{22}$ and $\ell_{11}<\ell_{22}$ we have
    \begin{equation}\label{eq:DelZer_iGj1}
    \sum_{\ell_{11}=0,2,\dots}^{\ell/2-2}\,\sum_{\ell_{22}=\ell_{11}+2}^{\ell-\ell_{11}}\msE^{1,\D_0}_{<}
    =\frac{1}{7680}\ell (\ell+4) (\ell+8) \left(2 \ell^2+11 \ell+24\right)
    \end{equation}
    while for $\ell<\ell_{11}+\ell_{22}$ the first term continues summing where the first sum in~\eqref{eq:DelZer_iGj1} left off summing. The second term continues the work of the second sum of~\eqref{eq:DelZer_iGj1}:
    \begin{align}\label{eq:DelZer_iGj2}
      &\sum_{\ell_{11}=\ell/2}^{\ell-2}\,\sum_{\ell_{22}=\ell_{11}+2}^{\ell}\msE^{2,\D_0}_{<}
      +\sum_{\ell_{11}=2}^{\ell/2-2}\,\sum_{\ell_{22}=\ell-\ell_{11}+2}^{\ell}\msE^{2,\D_0}_{<}
      =\frac{1}{7680}\ell (\ell+4) \left(14 \ell^3+69 \ell^2+224 \ell-16\right).
    \end{align}
    By summing~\eqref{eq:DelZer_iGj1} and~\eqref{eq:DelZer_iGj2} we arrive at
    \begin{equation}\label{eq:DelZer_iGj}
      \mathsf{e}^{\D_0}_{<}=\frac{1}{480} \ell (\ell+4) \left(\ell^3+6 \ell^2+21 \ell+11\right).
    \end{equation}
    When $\ell_{11}=\ell_{22}$ we obtain
    \begin{equation}\label{eq:Delzer_iEj1}
      \sum_{\ell_{11}=0,2,\dots}^{\ell/2}\msE^{1,\D_0}_{=}=\frac{1}{768} (\ell+4) (\ell+8) \left(\ell^2+8 \ell+24\right)
    \end{equation}
    and
    \begin{equation}\label{eq:Delzer_iEj2}
      \sum_{\ell_{11}=\ell/2+2}^{\ell}\msE^{2,\D_0}_{=}=\frac{1}{768} \ell (\ell+4) \left(7 \ell^2+32 \ell+120\right).
    \end{equation}
    Their sum equals
    \begin{equation}\label{eq:Delzer_iEj}
      \mathsf{e}^{\D_0}_{=}=\frac{1}{96} (\ell+4) \left(\ell^3+6 \ell^2+26 \ell+24\right).
    \end{equation}

    \subsection*{Case $4\mid\ell-1$ and $\D_-$}
    For $\D<0$ the inequalities $\ell\geq\ell_{11}+\ell_{22}-3$ and $\ell_{11}<\ell_{22}$
    \begin{align}\label{eq:DelNeg_iGj1odd}
      &\sum_{\ell_{22}=3,5,\dots}^\ell \msD^{1,\D_-}_{<}
      +\sum_{\ell_{11}=3,5,\dots}^{(\ell+1)/2}\,\sum_{\ell_{22}=\ell_{11}+2}^{\ell-\ell_{11}+3}\msD^{1,\D_-}_{<}\nn\\
      &=\frac{1}{46080}(\ell-1) \left(\ell^5+34 \ell^4+479 \ell^3+3509 \ell^2+14268 \ell+10125\right),
    \end{align}
    where in the first sum we set $\ell_{11}=1$ to be treated separately. For $\ell<\ell_{11}+\ell_{22}-3$ we get
    \begin{align}\label{eq:DelNeg_iGj2odd}
      &\sum_{\ell_{11}=(\ell+1)/2+2}^{\ell-2}\,\sum_{\ell_{22}=\ell_{11}+2}^{\ell}\msD^{2,\D_-}_{<}
      +\sum_{\ell_{11}=5}^{(\ell+1)/2}\,\sum_{\ell_{22}=\ell-\ell_{11}+5}^{\ell}\msD^{2,\D_-}_{<}\nn\\
      &=\frac{1}{46080}(\ell-5) (\ell-1) \left(19 \ell^4+261 \ell^3+1526 \ell^2+4221 \ell+2997\right).
    \end{align}
    Eqs.~\eqref{eq:DelNeg_iGj1odd} and~\eqref{eq:DelNeg_iGj2odd} sum to
    \begin{equation}\label{eq:DelNeg_iGJ}
      \mathsf{d}^{\D_-}_{<}=\frac{1}{2304}(\ell-1) (\ell+3) \left(\ell^4+7 \ell^3+14 \ell^2-37 \ell-81\right).
    \end{equation}
    For $\ell_{11}=\ell_{22}$ we obtain
    \begin{equation}\label{eq:DelNeg_iEj1odd}
      \sum_{\ell_{11}=3,5,\dots}^{(\ell+1)/2}\msD^{1,\D_-}_{=}
      =\frac{1}{7680}(\ell-1) (\ell+3) (\ell+7) \left(\ell^2+16 \ell+75\right)
    \end{equation}
    (note that $\ell_{11}=\ell_{22}=0$ is excluded since then only $\ell_{33}=\ell_{44}$ is admissible and so it belongs to the $\D=0$ case) and in the second case we continue summing:
    \begin{equation}\label{eq:DelNeg_iEJ2odd}
      \sum_{\ell_{11}=(\ell+1)/2+2}^{\ell}\msD^{1,\D_-}_{=}
      =\frac{1}{7680}(\ell-1) \left(23 \ell^4+258 \ell^3+1148 \ell^2+2038 \ell+885\right).
    \end{equation}
    The sum of~\eqref{eq:DelNeg_iEj1odd} and~\eqref{eq:DelNeg_iEJ2odd} equals
    \begin{equation}\label{eq:DelNeg_iEjodd}
      \mathsf{d}^{\D_-}_{=}=\frac{1}{1920}(\ell-1) (\ell+3) \left(6 \ell^3+53 \ell^2+192 \ell+205\right).
    \end{equation}

    \subsection*{Case $4\mid\ell-1$  and $\D_0$}
    For $\ell\geq\ell_{11}+\ell_{22}-1$ and $\ell_{11}<\ell_{22}$ we find
    \begin{equation}\label{eq:DelZer_iGj1odd}
    \sum_{\ell_{11}=1,3,\dots}^{(\ell+1)/2-2}\,\sum_{\ell_{22}=\ell_{11}+2}^{\ell-\ell_{11}+1}\msD^{1,\D_0}_{<}
    =\frac{1}{7680}(\ell-1) (\ell+3) \left(2 \ell^3+41 \ell^2+304 \ell+805\right).
    \end{equation}
    Similarly to $\ell$ even, for $\ell<\ell_{11}+\ell_{22}-1$ the first term continues summing where the first sum in~\eqref{eq:DelZer_iGj1odd} ended and the second term continues where the second sum of~\eqref{eq:DelZer_iGj1odd} terminated:
    \begin{align}\label{eq:DelZer_iGj2odd}
      &\sum_{\ell_{11}=(\ell+1)/2}^{\ell-2}\,\sum_{\ell_{22}=\ell_{11}+2}^{\ell}\msD^{2,\D_0}_{<}
      +\sum_{\ell_{11}=3}^{(\ell+1)/2-2}\,\sum_{\ell_{22}=\ell-\ell_{11}+3}^{\ell}\msD^{2,\D_0}_{<}\nn\\
      &=\frac{1}{7680}(\ell-1) (\ell+3) \left(14 \ell^3+87 \ell^2+208 \ell-165\right).
    \end{align}
    The sum of~\eqref{eq:DelZer_iGj1odd} and~\eqref{eq:DelZer_iGj2odd} equals
    \begin{equation}\label{eq:DelZer_iGjodd}
      \mathsf{d}^{\D_0}_{<}=\frac{1}{480} (\ell-1) (\ell+2) (\ell+3) \left(\ell^2+6 \ell+20\right)
    \end{equation}
    When $\ell_{11}=\ell_{22}$ we obtain
    \begin{equation}\label{eq:Delzer_iEj1odd}
      \sum_{\ell_{11}=1,3,\dots}^{(\ell+1)/2}\msD^{1,\D_0}_{=}=\frac{1}{768} (\ell+3) (\ell+7) \left(\ell^2+14 \ell+57\right)
    \end{equation}
    and
    \begin{equation}\label{eq:Delzer_iEj2odd}
      \sum_{\ell_{11}=(\ell+1)/2+2}^{\ell}\msD^{2,\D_0}_{=}=\frac{7}{768} (\ell-1) \left(\ell^3+9 \ell^2+35 \ell+51\right)
    \end{equation}
    with their sum being
    \begin{equation}\label{eq:Delzer_iEjodd}
      \mathsf{d}^{\D_0}_{=}=\frac{1}{96} (\ell+3) \left(\ell^3+7 \ell^2+29 \ell+35\right).
    \end{equation}

    Denoting $\D>0$ by $\D_+$, Lemma~\ref{lem:DioSimSystem} tells us that $\mathsf{e}^{\D_+}_{<}=\mathsf{e}^{\D_-}_{<},\mathsf{e}^{\D_+}_{=}=\mathsf{e}^{\D_-}_{=}$ and $\mathsf{d}^{\D_+}_{<}=\mathsf{d}^{\D_-}_{<},\mathsf{d}^{\D_+}_{=}=\mathsf{d}^{\D_-}_{=}$. Lemma~\ref{lem:ell1ell22} brings the solutions for $\ell_{11}>\ell_{22}$: $\mathsf{e}^{\D_-}_{>}=\mathsf{e}^{\D_-}_{<}, \mathsf{e}^{\D_+}_{>}=\mathsf{e}^{\D_+}_{<}$ and $\mathsf{e}^{\D_0}_{>}=\mathsf{e}^{\D_0}_{<}$. The same holds for odd $\ell$ and we get $\mathsf{d}^{\D_-}_{>}=\mathsf{d}^{\D_-}_{<},\mathsf{d}^{\D_+}_{>}=\mathsf{d}^{\D_+}_{<}$ and $\mathsf{d}^{\D_0}_{>}=\mathsf{d}^{\D_0}_{<}$. Eq.~\eqref{eq:ellEven} is obtained from
    \begin{equation}
      \mathsf{e}=4\mathsf{e}^{\D_-}_{<}+2\mathsf{e}^{\D_-}_{=}+2\mathsf{e}^{\D_0}_{<}+\mathsf{e}^{\D_0}_{=}
    \end{equation}
    and Eq.~\eqref{eq:ellOdd} from
    \begin{equation}
      \mathsf{d}=4\mathsf{d}^{\D_-}_{<}+2\mathsf{d}^{\D_-}_{=}+2\mathsf{d}^{\D_0}_{<}+\mathsf{d}^{\D_0}_{=}.
    \end{equation}
\end{proof}

\section{Secondary result}
If we set $\a_{i4}=0$ in~\eqref{eq:DiophantineSim} then it becomes a simpler linear system with an interesting number of nonnegative solutions.
\begin{proposition}\label{prop:Floyd}
  The number of nonnegative solutions of the following system of linear Diophantine equations
  \begin{subequations}\label{eq:Floyd}
  \begin{align}
    2\a_{11}+\a_{12}+\a_{13} & = \ell,\label{eq:Floyd1} \\
    \a_{12}+2\a_{22}+\a_{23} & = \ell,\label{eq:Floyd2} \\
    \a_{13}+\a_{23}+2\a_{33} & = \ell,\label{eq:Floyd3}
  \end{align}
  \end{subequations}
  is
  \begin{equation}\label{eq:FloydNoofSols}
    f=\frac{1}{16} (\ell+2) \left(\ell^2+4 \ell+8\right)
  \end{equation}
  for $\ell$ even and zero for $\ell$ odd.
\end{proposition}
\begin{remark}
  By setting $\ell\mapsto2\ell-2$ we get
  \begin{equation}\label{eq:magicSq}
    F(\ell)={1\over2}\ell(1+\ell^2).
  \end{equation}
  Eq.~\eqref{eq:magicSq} obtained after the rescaling of~\eqref{eq:FloydNoofSols} is the sum of rows, columns or diagonals of a normal magic square of the size $\ell>2$ (sometimes called the magic constant). It is also known as the sum of rows in Floyd's triangle.
\end{remark}
\begin{lemma}\label{lem:Floyd}
  Let $\ell_{ii}=\ell-2\a_{ii}\geq0$ for $i=1,2,3$. Then there exists a nonnegative solution of~\eqref{eq:Floyd} if and only if
  \begin{equation}\label{eq:IneqDiffSum}
    |\ell_{22}-\ell_{11}|\leq\ell_{33}\leq\ell_{11}+\ell_{22}.
  \end{equation}
\end{lemma}
\begin{proof}
  The direct part follows from summing any two of the three equations~\eqref{eq:Floyd} and subtracting the third one. We get three equations of the form
  \begin{equation}\label{eq:alphaij}
    2\a_{ij}=\ell_{ii}+\ell_{jj}-\ell_{kk}.
  \end{equation}
  Since we are looking for $\a_{ij}\geq0$ it is necessary the following to be true: $\ell_{11}+\ell_{22}\geq\ell_{33},\ell_{22}+\ell_{33}\geq\ell_{11}$ and $\ell_{11}+\ell_{33}\geq\ell_{22}$. The first inequality is the RHS of~\eqref{eq:IneqDiffSum} and combining the last two expressions we get the LHS. For the converse we may assume that~\eqref{eq:IneqDiffSum} is violated (the first or second inequality). Then from~\eqref{eq:alphaij} we immediately see that $\a_{ij}$ is negative.
\end{proof}
\begin{remark}
  Note that the second inequality in~\eqref{eq:IneqDiffSum} may not be saturated for some $\ell_{11},\ell_{22}$. That is, not all $\ell_{ii}$  satisfying~\eqref{eq:IneqDiffSum} are actually admissible. This will become relevant in the next proof.
\end{remark}
\begin{proof}[Proof of Proposition~\ref{prop:Floyd}]
  First we show that for $\ell$ odd there is no solution to~\eqref{eq:Floyd}. In that case $\ell_{ii}$ are odd as well and by plugging them to~\eqref{eq:alphaij} we always get the RHS to be an odd number. But then $\a_{ij}$ cannot be an integer. So from now on we focus on $\ell$ even. Since $0\leq\ell_{33}\leq\ell$, then from~\eqref{eq:IneqDiffSum} it follows that the number of non-empty nonnegative solutions of~\eqref{eq:Floyd} are determined by $\ell_{ii}$ satisfying
  \begin{equation}\label{eq:IneqDiffSumSaturated}
    |\ell_{22}-\ell_{11}|\leq\ell_{33}\leq\min{[\ell_{11}+\ell_{22},\ell]}.
  \end{equation}
  To count the number of solutions means to count the number of triples $(\ell_{ii})_{i=1}^3$ satisfying~\eqref{eq:IneqDiffSumSaturated}. This is because for any such triple we get a triple of $\a_{ij}$ via~\eqref{eq:alphaij}. Contrary to the original system~\eqref{eq:DiophantineSim} there are no multiplicities and this significantly simplifies the counting. The RHS of~\eqref{eq:IneqDiffSumSaturated} splits into two cases when $\ell\geq\ell_{11}+\ell_{22}$ and $\ell<\ell_{11}+\ell_{22}$. The counting of admissible triples $(\ell_{ii})_{i=1}^3$ in~\eqref{eq:IneqDiffSumSaturated} is invariant w.r.t. the relabelling $1\leftrightharpoons2$ and so the number of solutions for $\ell_{11}>\ell_{22}$ equals that of $\ell_{11}<\ell_{22}$. Thus, assuming $\ell_{11}\leq\ell_{22}$,  the number of solutions for the first case of~\eqref{eq:IneqDiffSumSaturated} is
  \begin{equation}\label{eq:FloydCaseI}
     m={1\over2}(\ell_{11}+\ell_{22}-(\ell_{22}-\ell_{11}))+1=\ell_{11}+1
  \end{equation}
  and in the second case it is
  \begin{equation}\label{eq:FloydCaseII}
   n={\ell-(\ell_{22}-\ell_{11})\over2}+1.
  \end{equation}
  In the geometric picture what we did is to count the number of points lying on a line between two endpoints. We divided by two because the values of $\ell_{ii}$ are a multiple of two and added one not to omit a boundary point.

  Two cases must be distinguished.
  \subsection*{Case $4\mid\ell$}
   We use the simple fact that $4\mid\ell\Rightarrow2\mid\ell$  in the following text. Considering the previously mentioned symmetry and~\eqref{eq:IneqDiffSumSaturated}, there are four cases to investigate in this section. For $\ell_{11}=\ell_{22}$ and $\ell\geq\ell_{11}+\ell_{22}$ we get from~\eqref{eq:FloydCaseI}
  \begin{equation}\label{eq:f1DivZero}
    f_1=\sum_{\ell_{11}=0,2,\dots}^{\ell/2}(\ell_{11}+1).
  \end{equation}
  Similarly, for $\ell_{11}=\ell_{22}$ and $\ell<\ell_{11}+\ell_{22}$ we find from~\eqref{eq:FloydCaseII}
  \begin{equation}\label{eq:f2DivZero}
    f_2=\bigg({\ell\over2}+1\bigg){\ell\over4}.
  \end{equation}
  The factor $\ell/4$ comes from finding the `smallest' solution for  $\ell_{11}=\ell_{22}$ and $\ell<\ell_{11}+\ell_{22}$ which is $\ell=-4+\ell_{11}+\ell_{22}=-4+2\ell_{22}^{\mathrm{min}}$. Then, the distance between $\ell_{22}^{\mathrm{max}}=\ell$ and $\ell_{22}^{\mathrm{min}}$ is $\ell-(\ell+4)/2$ giving us $\ell/4$ after dividing by two and adding one (a boundary point). For $\ell\geq\ell_{11}+\ell_{22}$ and $\ell_{11}<\ell_{22}$ we get
  \begin{equation}\label{eq:f3DivZero}
    f_3=\sum_{\ell_{11}=0,2,\dots}^{\ell/2}(\ell_{11}+1)\bigg({\ell\over2}-\ell_{11}\bigg).
  \end{equation}
  The second term counts the number of points lying between $\ell\geq\ell_{11}+\ell_{22}$ and $\ell_{11}=\ell_{22}+2$.

  The smallest $\ell_{22}$ consistent with the last case ($\ell<\ell_{11}+\ell_{22}$ and $\ell_{11}<\ell_{22}$) is $\ell_{22}=\ell/2+2$ (follows from minimally saturating the inequalities: $\ell+c_1=\ell_{11}+\ell_{22}$ and $\ell_{11}+c_2=\ell_{22}$ for $c_1=c_2=2$ and extracting $\ell_{22}$). Then, for every $\ell_{22}$ from $\ell/2+2$ to its maximal value  (equal to $\ell$) we find the corresponding $\ell_{11}$. This can be done in the following way. By solving for $\ell_{11}$ and inserting it to the second equation we get
  \begin{equation}\label{eq:c1c2}
    \ell_{22}={\ell+c_1+c_2\over2}
  \end{equation}
  ($c_1=c_2=2$ gives us the previous minimal saturation). Since $\ell_{22}$ increases by two, the closest allowed value after $c_1+c_2=4$ is $c_1+c_2=8$. By recalling $c_i\geq2$ it follows that there are now three possibilities: $(c_1,c_2)=\{(2,6),(4,4),(6,2)\}$ and for every increment of $c_1+c_2$ by four we add two more solutions. This gives us the necessary counting and taking into account~\eqref{eq:FloydCaseII} we may write
  \begin{equation}\label{eq:f4DivZero}
    f_4=\sum_{\ell_{22}=\ell/2+2}^\ell\,\sum_{\ell_{11}=\ell-\ell_{22}+2}^{\ell_{22}-2}\bigg({\ell-(\ell_{22}-\ell_{11})\over2}+1\bigg).
  \end{equation}
  The inner upper/lower bound is calculated from $\ell/2\pm(\ell_{22}-\ell/2-2)$. This expression follows after we find the minimal $\ell_{11}=\ell/2$ (again from setting $c_1=c_2=2$). Since the minimal $\ell_{22}$ equals $\ell/2+2$ and $\ell-\ell_{22}<\ell_{11}<\ell_{22}$ the sum's bounds follow. We get~\eqref{eq:FloydNoofSols} from $f=f_1+f_2+2f_3+2f_4$. The factors of two account for the number of solutions for $\ell_{11}>\ell_{22}$ which is equal to the studied case $\ell_{11}<\ell_{22}$.

  \subsection*{Case $4\mid\ell-2$}
  We again tacitly use $4\mid\ell-2\Rightarrow2\mid\ell-2$. The derivation is very similar so let us stress common points and noteworthy differences. Essentially, the main difference comes from the fact that the boundaries $\ell_{11}=\ell_{22}$ and $\ell=\ell_{11}+\ell_{22}$  (leading to the split to four cases) intersect at $(\ell/2,\ell/2)$ which is odd. But the values of $\ell_{11},\ell_{22}$ are never odd in our  problem and so it is mostly about adjusting the sums' bounds (to be shifted by one to start/end counting at an even point). So for $\ell_{11}=\ell_{22}$ and $\ell\geq\ell_{11}+\ell_{22}$ we now get
  \begin{equation}\label{eq:f1DivZerotilde}
    \tilde{f}_1=\sum_{\ell_{11}=0,2,\dots}^{\ell/2-1}(\ell_{11}+1)
  \end{equation}
  ($\ell_{11}=\ell/2-1$ is the last even point consistent with the inequalities). For  $\ell_{11}=\ell_{22}$ and $\ell<\ell_{11}+\ell_{22}$ the smallest solution is now $\ell=-2+\ell_{11}+\ell_{22}=-2+2\ell_{22}^{\mathrm{min}}$ and $(\ell_{22}^{\mathrm{max}}-\ell_{22}^{\mathrm{min}})/2+1=(\ell/2+1)/2$. Thus
  \begin{equation}\label{eq:f2DivZerotilde}
    \tilde f_2=\bigg({\ell\over2}+1\bigg)\bigg({\ell\over2}+1\bigg){1\over2}.
  \end{equation}
  In the case $\ell\geq\ell_{11}+\ell_{22}$ and $\ell_{11}<\ell_{22}$ the counting argument goes through exactly like for~\eqref{eq:f3DivZero} except that in order to satisfy the inequalities we stop the counting of $\ell_{11}$ on the last even number (which is $\ell/2-1$)
  \begin{equation}\label{eq:f3DivZerotilde}
    \tilde f_3=\sum_{\ell_{11}=0,2,\dots}^{\ell/2-1}(\ell_{11}+1)\bigg({\ell\over2}-\ell_{11}\bigg).
  \end{equation}
  Finally, for $\ell<\ell_{11}+\ell_{22}$ and $\ell_{11}<\ell_{22}$, we get~\eqref{eq:c1c2} as well but to get to the closest admissible even $\ell_{22}$ for $4\mid\ell-2$ we have to shift it by one:
  \begin{equation}\label{eq:c1c2tilde}
    \ell_{22}={\ell+c'_1+c'_2\over2}+1={\ell+c_1+c_2\over2}.
  \end{equation}
  Recalling $c'_i\geq2$, we get two minimally saturating solutions ($(c_1,c_2)=\{(2,4),(4,2)\}$) and as before, by increasing $\ell_{22}$ by two, two more solutions are always added. Hence
  \begin{equation}\label{eq:f4DivZerotilde}
    \tilde f_4=\sum_{\ell_{22}=\ell/2+3}^\ell\,\sum_{\ell_{11}=\ell-\ell_{22}+2}^{\ell_{22}-2}\bigg({\ell-(\ell_{22}-\ell_{11})\over2}+1\bigg),
  \end{equation}
  where the outer lower bound comes from the RHS of~\eqref{eq:c1c2tilde} for the minimal solutions. For the inner bounds we get the same expressions like in~\eqref{eq:f4DivZero} but the derivation is modified by realizing that there are two minimal solutions $\ell_{11}=\ell/2\pm1$ for the minmal $\ell_{22}=\ell/2+3$.  We again get~\eqref{eq:FloydNoofSols} from $\tilde f_1+\tilde f_2+2\tilde f_3+2\tilde f_4$.
\end{proof}

\section{Discussion and open problems}\label{sec:open}

We conclude this work with several remarks. We found a closed expression counting the number of nonnegative solutions of linear
Diophantine system of equations~\eqref{eq:Diophantine} for $\ell_i=\ell\geq0$ and for its special case of $\a_{i4}=0$. The main linear system is motivated by counting the perturbative contributions for an interaction Lagrangian in interacting quantum field theory for bosons. In particular, the number of nonnegative solutions is closely related to counting the Feynman diagrams for two interacting fields in the scalar $\phi^n$ model to an arbitrary perturbative order~\cite{bradler2016unitary}.

It would be quite interesting to generalize the presented result to a linear system given by generalizing~\eqref{eq:Diophantine} in the following way:
\begin{equation}\label{eq:DiphantGeneral}
  2\a_{ii}+\sum_{\genfrac{}{}{0pt}{2}{j=1}{j\neq i}}^k\a_{ij}=\ell_{i}
\end{equation}
for $1\leq i\leq k$. This would provide a very general counting method of Feynman diagrams for an arbitrary number of interacting fields, to an arbitrary perturbative order and for any scalar $\phi^n$ model of interacting bosons.

The secondary problem was motivated purely by curiosity as what happens if we simplify the Diophantine system and has no bearing to high-energy physics. Unexpectedly, after  rescaling $\ell\mapsto2\ell-2$ and for $\ell$ even, the number of nonnegative solutions of such a system (Eq.~\eqref{eq:Floyd}) turn out to be the magic constant of order $\ell$ -- the sum of all rows, columns and diagonals of a normal magic square of order $\ell>2$.

\bibliographystyle{unsrt}


\end{document}